\documentclass[11pt]{article}
\usepackage[T1]{fontenc}
\usepackage{lmodern}
\usepackage{geometry} 
\usepackage{graphicx}
\usepackage{amsmath,amssymb,amsthm,amsfonts,dsfont,mathtools}
\usepackage{relsize}
\usepackage[margin=1cm]{caption}
\usepackage{wrapfig}
\usepackage{frame,color}
\usepackage{environ,subfig}
\usepackage{xspace}
\usepackage{framed}
\usepackage{comment}
\usepackage{changepage}
\usepackage{enumerate}
\usepackage{algorithm}
\usepackage[noend]{algpseudocode}
\usepackage{algorithmicx}
\usepackage{mathrsfs}
\usepackage{soul}
\usepackage{authblk}
\usepackage{hyperref}
\usepackage[utf8]{inputenc}
\usepackage{tikz}
\usepackage{microtype,hyperref}
   \hypersetup{%
      breaklinks,%
      ocgcolorlinks, colorlinks=true,%
      urlcolor=[rgb]{0.25,0.0,0.0},%
      linkcolor=[rgb]{0.5,0.0,0.0},%
      citecolor=[rgb]{0,0.2,0.445},%
      filecolor=[rgb]{0,0,0.4},
      anchorcolor=[rgb]={0.0,0.1,0.2}%
   }
\usepackage{url}
\usepackage{breakurl}

\usepackage{lineno}

\newcommand{\PAPER}[1]{}

\algnewcommand{\Inputs}[1]{%
  \State \textbf{Inputs:}
  \Statex \hspace*{\algorithmicindent}\parbox[t]{.8\linewidth}{\raggedright #1}
}
\algnewcommand{\Initialize}[1]{%
  \State \textbf{Initialize:}
  \Statex \hspace*{\algorithmicindent}\parbox[t]{.8\linewidth}{\raggedright #1}
}
\algnewcommand{\TurnOne}[1]{%
  \State \textbf{Timestep 1:}
  \Statex \hspace*{\algorithmicindent}\parbox[t]{.8\linewidth}{\raggedright #1}
}

\makeatletter
 {\par\unskip\endMakeFramed}
\makeatother

\makeatletter
\newsavebox{\@brx}
\newcommand{\llangle}[1][]{\savebox{\@brx}{\(\m@th{#1\langle}\)}%
  \mathopen{\copy\@brx\mkern2mu\kern-0.9\wd\@brx\usebox{\@brx}}}
\newcommand{\rrangle}[1][]{\savebox{\@brx}{\(\m@th{#1\rangle}\)}%
  \mathclose{\copy\@brx\mkern2mu\kern-0.9\wd\@brx\usebox{\@brx}}}
\makeatother

\NewEnviron{fcodebox}[1]%
{\fbox{%
\begin{minipage}{#1\linewidth}
\begin{codebox}
\BODY
\end{codebox}
\end{minipage}
}}

\newtheorem{fact}{Fact}

\newtheorem{theorem}[fact]{Theorem}
\newtheorem{lemma}[fact]{Lemma}
\newtheorem{conjecture}[fact]{Conjecture}

\newcommand{\ignore}[1]{}

\newcommand{\Cdense}{C_{dense}}
\newcommand{\Alphaave}{\alpha_{ave}}

\newcommand{\tO}{\widetilde{O}}

\newcommand{\polylog}{\mathop{\mathrm{polylog}}}

\newcommand{\eps}{\varepsilon}

\makeatletter
\DeclareRobustCommand\onedot{\futurelet\@let@token\@onedot}
\def\@onedot{\ifx\@let@token.\else.\null\fi\xspace}

\makeatother

\title{Improved Algorithms for Integer Complexity}

\author[1]{Qizheng He\thanks{qizheng6@illinois.edu.}}
\affil[1]{Department of Computer Science, University of Illinois at Urbana-Champaign}

\begin{document}
\date{}
\maketitle
\thispagestyle{empty}

\begin{abstract}
The integer complexity $f(n)$ of a positive integer $n$ is defined as the minimum number of 1's needed to represent $n$, using additions, multiplications and parentheses. We present two simple and faster algorithms for computing the integer complexity:

1) A near-optimal $O(N\polylog N)$-time algorithm for computing the integer complexity of all $n\leq N$, improving the previous $O(N^{1.223})$ one [Cordwell et al., 2017].

2) The first sublinear-time algorithm for computing the integer complexity of a single $n$, with running time $O(n^{0.6154})$. The previous algorithms for computing a single $f(n)$ require computing all $f(1),\dots,f(n)$.
\end{abstract}


\section{Introduction}
The \emph{integer complexity} $f(n)$ of a positive integer $n$ is a simple-looking problem in number theory with a long history. It is defined as the minimum number of 1's needed to represent $n$, using basic arithmetic expressions that only include 1, additions, multiplications and parentheses. For example, $f(6)=5$, since $6=(1+1)\cdot (1+1+1)$. Due to the simplicity of its description, this problem is a popular exercise in undergraduate algorithm courses.

The problem of integer complexity was first implicitly posed by Mahler and Popken in 1953~\cite{mahler1953maximum}, and then explicitly stated by Selfridge~\cite{guy1986some}. It was later popularized by Guy, who included this problem in his book ``Unsolved Problems in Number Theory''~\cite{guy2004unsolved}. Despite the simplicity of its problem definition, a considerable amount of papers has studied this problem from both number-theoretical and algorithmic perspectives (see OEIS A005245~\cite{A005245} for a brief summary), yet the problem has not been satisfactorily solved.

\paragraph{Number-theoretical results.} On the number-theory side, mathematicians mainly work towards proving tighter asymptotic upper and lower bounds on the integer complexity.

For the upper bound, let $\alpha$ denote $\sup_{n>1}\frac{f(n)}{\log_3 n}$, so that $f(n)\leq \alpha\log_3 n$ holds for all $n>1$. Guy~\cite{guy1986some} noted a trivial result $f(n)\leq 3\log_2 n$ by writing $n$ in base 2, so $\alpha\leq 4.755$.
Surprisingly, this simple bound had not been improved for a long time, until Zelinsky proved a better upper bound $\alpha\leq 4.676$~\cite{Zelinsky09post}, and then further improved to $\alpha \leq 41\log_{55296} 3\leq 4.125$~\cite{zelinsky2022upper}.
\begin{lemma}[\cite{zelinsky2022upper}]\label{lemma:upper_bound}
$f(n)\leq \alpha\log_3 n$ for all positive integers $n>1$, where $\alpha\leq 4.125$.
\end{lemma}
People also studied upper bounds on the integer complexity for almost all numbers except for a set of density zero~\cite{guy1986some,de2014algorithms,steinerberger2014short}, and the current best bound is $3.620\log_3 n$ by Cordwell et al.~\cite{cordwell2017algorithms}. Shriver~\cite{shriver2015application} proved an upper bound $3.52\log_3 n$ for numbers with logarithmic density one.

For the lower bound, a simple result $f(n)\geq 3\log_3 n$ is attributed to Selfridge~\cite{guy1986some}, which is tight when $n$ is a power of $3$. Altman and Zelinsky studied integers whose complexity is close to this lower bound~\cite{altman2012numbers}.

For more related number-theoretic works, see \cite{altman2012numbers,altman2014integer,altman2016IntegerBounded,10.2140/moscow.2019.8.193,complexity2021}.

\paragraph{Algorithmic results.} On the algorithms side, for computing the integer complexity $f(n)$ of a single target $n$, it is a standard homework exercise in undergraduate courses to design a naive dynamic programming (DP) algorithm that runs in $O(n^2)$ time, using the recurrence
$$f(n)=\min\left\{\min_{1\leq i\leq n-1}\left(f(i)+f(n-i)\right),\min_{m:m|n}\left(f(m)+f(n/m)\right)\right\}.$$
Namely, if the last evaluated operation for getting $n$ has the form $n=i+(n-i)$ using addition, then $f(n)=f(i)+f(n-i)$. Otherwise if it has the form $n=m\cdot (n/m)$ using multiplication, then $f(n)=f(m)+f(n/m)$.

A key observation by Vivek and Shankar~\cite{venkatesh2008integer} is, for the addition case (which is the bottleneck of the DP algorithm), it suffices to only consider the indices $1\leq i\leq L(n)\triangleq n^{\ell}$ for some constant exponent $\ell$: if $i$ is too large (w.l.o.g.\ assume $i\leq n-i$), then the total integer complexity of the two parts $f(i)+f(n-i)\geq 3\log_3 i+3\log_3 (n-i)$ will be larger than the upper bound $\alpha\log_3 n$ on $f(n)$, thus cannot be optimal. By solving the equation $3\ell\log_3 n+3\log_3 (n-n^\ell)=\alpha\log_3 n$, and using the current best upper bound on $\alpha$ (Lemma~\ref{lemma:upper_bound}), we get $\ell \approx\alpha/3-1\leq 0.3749$. For the multiplication case, since the total number of factors of $1,\dots,n$ is only $O(n\log n)$, one can handle this case more efficiently, in only $O(n\log n)$ time. Combing this observation with the naive DP, they obtained an algorithm with running time $O(n^{\log_2 3})\approx O(n^{1.585})$ \cite{venkatesh2008integer}, using the trivial value $4.755$ for $\alpha$ available at that time. Two concurrent works~\cite{Fuller08code,de2014question} used similar ideas and were superior to~\cite{venkatesh2008integer}, but they did not analyze the running time of their algorithms.

These algorithms were further improved by J.\ Arias de Reyna and J.\ Van de Lune~\cite{de2014algorithms}, to running time $O(n^{1.231})$. The key idea is it suffices to use the better upper bounds available for integers with density $1$, and then handle the few exceptions, instead of using the upper bound $\alpha\log_3 n$ that holds for all $n$. Further improving on their idea, the current best algorithm by Cordwell et al.~\cite{cordwell2017algorithms} runs in $O(n^{1.223})$ time.

We note that in order to compute $f(N)$ for a \emph{single} target $N$, all these algorithms also simultaneously computed $f(n)$ for \emph{all} targets $n\leq N$. For the all-targets version, there is still a significant gap between their superlinear running time and the trivial $\Omega(N)$ lower bound. In this paper, we close this gap by providing a near-linear time algorithm. On the other hand, we will design faster sublinear-time algorithms that specifically work for a single target.

We remark that the two viewpoints, number-theoretic and algorithmic, are closely related. The process of designing faster algorithms benefit from the number-theoretic bounds; on the other hand, performing experiments using faster algorithms help for proposing, proving and disproving number-theoretic conjectures. For example, Iraids et al.~\cite{iraids2012integer} computed $f(n)$ up to $n\leq 10^{12}$, and presented their experimental observations. For more related algorithmic and computational works, see~\cite{vcercnenoks2015integer,altman2016integer}.


\paragraph{Overview of our new algorithms.}
We mainly study the integer complexity problem from the algorithmic side, and present simple and faster algorithms for both the all-targets version and the single-target version. In Sec.~\ref{sec:alg for all N}, we present a near-optimal $\tO(n)$\footnote{The $\tO$ notation hides polylogarithmic factors.}-time algorithm for computing $f(n)$ for all $n\leq N$, where the main idea is to use online $(\min,+)$-convolution with bounded entries to handle the addition case. This approach is significantly different from the previous works, which all use number-theoretic bounds to limit the range in addition, but then handle the additions in a brute-force way. In Sec.~\ref{sec:alg for single N}, we present the first sublinear-time algorithm for computing a single $f(n)$, where the key idea is to use a special recursion based on both the number-theoretic bounds for addition and observations on the special structure for multiplication, so that computing the current value only depends on a few previous values.

\section{Algorithm for All Targets}\label{sec:alg for all N}

\paragraph{Preliminaries.} We first introduce some preliminary definitions and tools that will be useful later for our algorithms.
\begin{lemma}[(classical) convolution]
The convolution of two arrays $a[0,\dots,n_1-1]$ and $b[0,\dots,n_2-1]$ is an array $c[0,\dots,n_1+n_2-2]$, where $c[i]=\sum_{j=0}^i a[j]\cdot b[i-j]$ (out-of-boundary elements are treated as 0).

It is known that convolution can be computed in $O((n_1+n_2)\log(n_1+n_2))$ time, by FFT.
\end{lemma}

\begin{lemma}[online convolution]\label{lemma:online_conv}
The online convolution of two arrays $a[0,\dots,n-1]$ and $b[0,\dots,n-1]$ is the same as the convolution of $a$ and $b$, but $a[i+1]$ and $b[i+1]$ are given in an online way, determined only after knowing the value of $c[i]=\sum_{j=0}^i a[j]\cdot b[i-j]$. Initially only $a[0]$ and $b[0]$ are known.

Hoeven~\cite{Hoeven02a,Hoeven07c} proved that online convolution can be computed in $\tO(n)$ time, by recursively performing convolutions in a divide and conquer way.
\end{lemma}

\begin{lemma}[$(\min,+)$-convolution with bounded entries]\label{lemma:min_plus_conv}
The $(\min,+)$-convolution of two arrays $a[0,\dots,n_1-1]$ and $b[0,\dots,n_2-1]$ is an array $c[0,\dots,n_1+n_2-2]$, where $c[i]=\min_{j=0}^i (a[j]+b[i-j])$ (out-of-boundary elements are treated as $\infty$).

It is known that if all entries in $a$ and $b$ are non-negative integers bounded by $u$, then the $(\min,+)$-convolution of $a$ and $b$ can be computed in $\tO(nu)$ time.
\end{lemma}

\begin{proof}
The algorithm is folklore, see e.g.~\cite{ChanL15}. The idea is by performing FFT on large numbers with $\tilde{O}(u)$ bits.
\end{proof}

\begin{lemma}[batched factoring]
We can compute the list of factors $D_i$ for all integers $i=1,\dots,n$ in $O(n\log n)$ time.
\end{lemma}

\begin{proof}
This is well-known, but for completeness we include a short proof here. We compute all lists of factors $D_i$ as follows. For $i=1,\dots,n$, enumerate each multiple $k=j\cdot i$ of $i$ (where $1\leq j\leq \lfloor\frac{n}{i}\rfloor$), then add $i$ to the list $D_k$. The running time is $O\left(\sum_{i=1}^{n} \frac{n}{i}\right)=O(n\log n)$.
\end{proof}

\paragraph{Main algorithm.} Now we explain the key ideas to achieve near-linear time for computing $f(n)$ for all $n\leq N$. The DP formula for handling multiplications only takes $O(\sum_{i=1}^N \sigma(i))=O(N\log N)$ overall time, where $\sigma(i)$ denotes the number of factors of $i$.

The current bottleneck is addition. Notice that if we only focus on additions, the DP formula $f(n)=\min_{1\leq i\leq n-1}\left(f(i)+f(n-i)\right)$ is actually an online $(\min,+)$-convolution. Furthermore, $f(n)=O(\log n)$ by Lemma~\ref{lemma:upper_bound}, so we are dealing with an online $(\min,+)$-convolution with entries bounded by $u=O(\log n)$. Combing the techniques of Lemma~\ref{lemma:online_conv} (which analogously also works for online $(\min,+)$-convolutions) and Lemma~\ref{lemma:min_plus_conv} yield the result.

A pseudocode of our all-targets algorithm is shown as follows:
\begin{algorithm}[!htbp]
\caption{Algorithm for computing $f(n)$ for all $n\leq N$}
\label{alg:all_targets}
\begin{algorithmic}[1]
\State Precompute the list of factors $D_i$ for each integer $i=1,\dots,N$.
\State Initialize $f(1)=1$ and $f(2),\dots,f(N)=\infty$.
\Function{func}{$l$, $r$}
    \If{$l=r$}
        \State Update $f(l)$ using multiplication: $f(l)=\min\left\{f(l),\min_{i\in D_l} (f(i)+f(l/i))\right\}$.
        \State Return
    \EndIf
    \State Set $m=(l+r)/2$.
    \State \Call{func}{$l$, $m$}
    \State Update $f(m+1),\dots,f(r)$, by computing the $(\min,+)$-convolution $[f'(m+1),\dots,f'(r)]$ of $[f(l),\dots,f(m)]$ and $[f(1),\dots,f(\min\{r-l,m\})]$ (appropriately truncated).
    \State \Call{func}{$m+1$, $r$}
\EndFunction
\State \Call{func}{$1$, $N$}
\end{algorithmic}
\end{algorithm}

The correctness proof is similar to the one for Lemma~\ref{lemma:online_conv}, with a simple modification for handling the multiplications. For completeness we provide the details here.
\begin{proof}
Suppose that for all $n\leq n_0$, $f(n)$ is computed correctly after the function \textsc{func}$(n,n)$ returns. We now prove that for $n=n_0+1$, $f(n)$ is also computed correctly after the function \textsc{func}$(n,n)$ returns. If $f(n)$ is obtained by multiplication ($f(n)=f(i)+f(n/i)$, $i>1$), then $f(n)$ will be computed correctly when we execute line 5 of \textsc{func}$(n,n)$, since $f(i)$ and $f(n/i)$ are both already computed correctly by the induction hypothesis. Otherwise if $f(n)$ is obtained by addition ($f(n)=f(i)+f(n-i)$, w.l.o.g.\ assume $i\leq n-i$), then consider the LCA node \textsc{func}$(l,r)$ of the two leaves \textsc{func}$(n-i,n-i)$ and \textsc{func}$(n,n)$ in the recursion tree. After performing \textsc{func}$(l,m)$, both $f(i)$ and $f(n-i)$ are already computed correctly (since $i\leq n-i\leq m<n\leq r$), and $f(n)$ will be updated by $f(i)+f(n-i)$ during the $(\min,+)$-convolution (since $l\leq n-i\leq m$, and $i\leq n-(n-i)\leq r-l$).
\end{proof}

The running time of the algorithm (apart from line 5) satisfies the recurrence $T(n)=2T(n/2)+\tO(n)$, which solves to $T(N)=\tO(N)$. Handling multiplications in line 5 takes overall $\tO(N)$ time.

\begin{theorem}\label{thm:algo_all}
Given an integer $N$, there exists an algorithm that can compute $f(n)$ for all $n\leq N$ in $O(N\polylog N)$ time.
\end{theorem}

\section{Algorithm for a Single Target}\label{sec:alg for single N}
To compute a single $f(n)$, we note the algorithms in previous works actually compute $f(n_0)$ for all targets $n_0\leq n$ (e.g., the $O(n^2)$ naive DP, and \cite{venkatesh2008integer,Fuller08code,de2014question,de2014algorithms,cordwell2017algorithms}). Here we show that if one is only interested in computing a single $f(n)$, then it is possible to spend only sublinear time.

\paragraph{A special recursion.} The idea is based on a special recursion as follows: say $f(n)$ is obtained by addition as the last operation. Then we can decompose $f(n)$ into $f(i)+f(n-i)$, where $1\leq i\leq L(n)$ (here $L(n)=n^\ell$ is an upper limit defined earlier, and $\ell=\alpha/3-1$), and $f(n-i)$ is obtained by multiplication (if $f(n)$ is obtained by adding multiple terms that sum to $n$, we let $n-i$ be the largest term. The observation of \cite{venkatesh2008integer} guarantees that $i\leq L(n)$). Let $f'(n)$ denote the integer complexity of $n$ where the last operation is required to be a multiplication (the special case is $f'(1)=1$). Then we want to know the $f'$ values in the range $[n-L(n),n)$, in order to get the correct value for $f(n-i)$. Those $f'$ values recursively depend on the $f$ values in intervals $[(n-L(n))/d,n/d]$ for some integer $d\geq 2$, when the $f'$ values are obtained by multiplying $d$ with another integer.

The key observation is the length of such intervals we want to compute converges to at most $2L(n)$: if we want to know the $f$ values in $[n/d_1-L(n),n/d_1]$ for some integer $d_1$, we only need the $f'$ values in $[n/d_1-2L(n),n/d_1]$, which further depend on the $f$ values in $[(n/d_1-2L(n))/d_2,n/d_1/d_2]$ for all possible integer $d_2\geq 2$. This is a subinterval of $[n/d-L(n),n/d]$ where $d=d_1d_2$.

We can now design our algorithm as follows. Let $t$ be a parameter to be chosen later. First use the algorithm in Sec.~\ref{sec:alg for all N} to precompute $f(n_0)$ and $f'(n_0)$ for all $n_0\leq n^t$ in $\tO(n^t)$ time. Then recursively compute the $f$ values in all intervals of the form $[n/d-L(n),n/d]$ for some integer $d\geq 1$ (and $f'$ values in all intervals of the form $[n/d-2L(n),n/d]$). For large enough $d$ such that $n/d\leq n^t$, we know that the $f$ and $f'$ values in the interval has already been precomputed. So we only need to compute $n^{1-t}$ such intervals.

We note that recursions of similar forms have appeared earlier for solving other problems, e.g., for unbounded knapsack and change-making \cite{AxiotisT19,ChanH22}, although they only require one recursive call on $d_1=2$, while ours are more involved, requiring several recursive branches.

For handling the multiplications, we can use Dixon's method~\cite{dixon1981asymptotically} for factoring an integer $n$ in $e^{O(\sqrt{\log n\log\log n})}=n^{o(1)}$ time, and then use the DP formula for multiplication.

\ignore{
\paragraph{Factoring.} Here we discuss different approaches for integer factoring, and analyze their running time.

1) We can use \href{https://en.wikipedia.org/wiki/Dixon\%27s\_factorization_method}{Dixon's method} with running time $e^{O(\sqrt{\log n\log\log n})}=n^{o(1)}$ (with a rigorous proof for the time complexity) to factor a single number at most $n$, which is fast enough.

2) An alternative is to use \cite[Theorem 3.8]{bringmann2021near} to handle all numbers that we want to factor in a batch, which takes $\tilde{O}(\sqrt{n})$ time in addition to the overall number of factors (and is not the bottleneck for the current algorithm).

3) For the current bound of $L(n)$, similar to the simple factorization method in Sec.~\ref{sec:alg for all N}, we can just factor each interval with right endpoint $n/d$ in $\tilde{O}(\sqrt{n/d}+L(n/d))=\tilde{O}(L(n/d))$ time, by enumerating all multiples of $i$ in the interval for $i=2,\dots,\sqrt{n/d}$. (The running time is $\sum_{i=2}^{\sqrt{n/d}}\lceil \frac{L(n/d)}{i}\rceil \leq \sum_{i=2}^{\sqrt{n/d}}(\frac{L(n/d)}{i}+1)\leq L(n/d)\log n+\sqrt{n/d}$.)

todo: write more on this.

update: compute multiples instead.
}

\paragraph{Running time analysis.} Now we analyze the running time of our algorithm. After precomputing the $f(n_0)$ and $f'(n_0)$ values for all $n_0\leq n^t$ in $\tO(n^t)$ time, we need to compute the $f$ and $f'$ values in all intervals of the form $[n/d-\Theta(L(n)),n/d]$ where $1\leq d\leq n^{1-t}$. To compute the $f$ values within one interval $[n/d-L(n),n/d]$ from $f'$, we use Lemma~\ref{lemma:min_plus_conv} for $(\min,+)$-convolution with entries bounded by $O(\log n)$, which takes $\tO(L(n))$ time. To compute the $f'$ values within an interval $[n/d-2L(n),n/d]$ from $f$, for all $i\in [n/d-2L(n),n/d]$, we enumerate all factors $d_1$ of $i$ and use the DP formula for multiplication. Using the upper bound $\sigma(i)=2^{O(\log i/\log\log i)}$ for the maximum number of factors of $i$, this step takes $O(\sum_{i\in [n/d-2L(n),n/d]}\sigma(i))=L(n)\cdot n^{o(1)}$ time.

The total running time is $\tO(n^t+n^{1-t}\cdot L(n)\cdot n^{o(1)})$. Set $t=\alpha/6$ to balance the terms, we get $n^{\alpha/6+o(1)}$ total running time.

\begin{theorem}
There exists an algorithm that can compute a single $f(n)$ in $O(n^{\alpha/6+o(1)})\leq O(n^{0.6875})$ time and space, where $\alpha=\sup_{n>1}\frac{f(n)}{\log_3 n}$.
\end{theorem}

\paragraph{Further improvements.} To further improve the running time, the idea is to use a tighter bound $L(n/d)$ for intervals with right endpoint $n/d$, instead of using the global upper bound $L(n)$. A pseudocode of our improved single-target algorithm is shown in Alg.~\ref{alg:single_target}.
\begin{algorithm}[!htbp]
\caption{Algorithm for computing a single $f(n)$}
\label{alg:single_target}
\begin{algorithmic}[1]
\State Let $n_0=n^t$. Compute $f(i)$ for all $i\leq n_0$, by Alg.~\ref{alg:all_targets}.
\For {$d=n/n_0,\dots,1$}
    \State Set $L_0=4L(n/d)$.
    \For {$i=n/d-L_0,\dots,n/d$}
        \State $f(i)=\min\{f(i),\min_{j:j|i}(f(j)+f(i/j))\}$.
    \EndFor
\State Update $f(n/d-L_0),\dots,f(n/d)$, by computing the $(\min,+)$-convolution of $[f(n/d-L_0),\dots,f(n/d)]$ and $[f(1),\dots,f(L_0)]$.
\EndFor
\State Return $f(n)$.
\end{algorithmic}
\end{algorithm}

If we want to know the $f$ values in $[n/d_1-c\cdot L(n/d_1),n/d_1]$ for some integer $d_1$ (where the constant $c$ is to be determined), we need the $f'$ values in $[n/d_1-(c+1)\cdot L(n/d_1),n/d_1]$, which further depend on the $f$ values in $[(n/d_1-(c+1)\cdot L(n/d_1))/d_2,n/d_1/d_2]$ for all possible integer $d_2\geq 2$. This is a subinterval of $[n/d-c\cdot L(n/d),n/d]$ where $d=d_1d_2$, using the inequality $(c+1)\cdot (n/d_1)^\ell/d_2\leq c\cdot (n/(d_1d_2))^\ell$, i.e., $c+1\leq c\cdot d_2^{1-\ell}$ that holds for all $d_2\geq 2$, if we set $c=3$. 

The new total running time is
\[
\tO\left(n^t+n^{o(1)}\cdot \sum_{d=1}^{n^{1-t}} (n/d)^{\ell}\right)=\tO\left(n^t+n^{o(1)}\cdot n^\ell \cdot n^{(1-t)\cdot (1-\ell)}\right).
\]
Set $t=\frac{1}{2-\ell}$ to balance the terms, the total running time is $n^{3/(9-\alpha)+o(1)}$.

\begin{theorem}\label{thm:algo_single}
There exists an algorithm that can compute a single $f(n)$ in $O(n^{3/(9-\alpha)+o(1)})\leq O(n^{0.6154})$ time and space, where $\alpha=\sup_{n>1}\frac{f(n)}{\log_3 n}$.
\end{theorem}
It is not hard to slightly modify our algorithm, so that the stated bounds in Theorem~\ref{thm:algo_all} and Theorem~\ref{thm:algo_single} also hold when replacing $\alpha$ with $\alpha_0$, where $\alpha_0=\limsup_{n\rightarrow \infty}\frac{f(n)}{\log_3 n}$, so that $f(n)\leq \alpha_0\log_3 n$ holds for all but finitely many $n$.

We remark that some ingredients in our algorithm look similar to Fuller's algorithm as mentioned in \cite{de2014algorithms}, but the algorithm there is used to save space when computing all $f(n)$. Our main contribution is to observe that this recursion is helpful for the \emph{single-target} problem, requiring knowing only a few previous values in order to compute the current value, and present a nontrivial running time analysis for this simple recursive algorithm.

As a side comment, although to achieve the best running time we require the $(\min,+)$-convolution ideas from Sec.~\ref{sec:alg for all N}, even using the recursion idea alone is sufficient to obtain a non-trivial improvement to the previous methods.

\paragraph{Experiments.} We implemented our faster single-target algorithm, and performed experiments to search for counterexamples for several number-theoretical conjectures in a larger range. For more details, see Appendix~\ref{sec:experiments}.

\section{Conclusion}
We conclude with a number of open problems:
\begin{itemize}
\item Can one prove a tighter upper bound for integer complexity? It is still open that whether the constant coefficient $\alpha_0$ can be as small as $3+\eps$. The running time of our single-target algorithm will be improved as long as one prove a better bound on $\alpha$ (or $\alpha_0$).
\item Can the integer complexity of $n$ be computed in polylogarithmic time? Currently no hardness results are known for this problem.
\end{itemize}

\paragraph{Acknowledgement.} We thank Timothy M.\ Chan, Ruta Mehta and Sariel Har-Peled for their CS374 undergraduate courses: the integer complexity problem was used as teaching material on the topic of dynamic programming. We thank Harry Altman for helpful discussions.

\bibliographystyle{plain}
\bibliography{references}

\appendix

\section{Experiments}\label{sec:experiments}
With more efficient algorithms, one can search for counterexamples for number-theoretic conjectures in a larger range. The up-to-date largest scale experiment is by Iraids et al.~\cite{iraids2012integer,vcercnenoks2015integer}, who computed all $f(n)$ up to $n\leq 10^{12}$, and presented their experimental observations.

We note that single-target algorithms are advantageous when verifying specific conjectures that only involve a sparse set of numbers. Nevertheless, most previous works used all-targets algorithms to perform the experiments. Therefore, with our new sublinear-time single-target algorithm in Sec.~\ref{sec:alg for single N}, we are able to perform experiments in a scale more than $10^{19}$ times larger, simply on a laptop.

We implemented our single-target algorithm in C++\footnote{The codes are available at \url{https://github.com/hqztrue/integer_complexity}.}. Some details are slightly modified for practical considerations (e.g., implementation of the 128-bit factoring algorithms, choosing the upper limit $L(n)$ in an instance-dependent way, computing lower bounds on the integer complexities~\cite{10.2140/moscow.2019.8.193} for pruning, and using brute-force implementation of $(\min,+)$-convolution instead of FFT-based), but the key idea of recursion remain the same.


\subsection{Conjectures on the integer complexity of $2^i$, $2^i3^j5^k$ and $2^i+1$}\label{sec:conj_power_of_2}

Rawsthorne~\cite{rawsthorne1989many,guy2004unsolved} posed the question that whether $f(2^i)=2i$ for all $i\geq 1$, which is viewed as one of the major open problems in this area~\cite{guy2004unsolved,iraids2012integer,altman2012numbers,de2014question,steinerberger2014short,vcercnenoks2015integer,zelinsky2022upper}. For example, if this hypothesis is proved, then an immediate corollary is there exists a constant $\eps>0$ such that $f(n)\geq (3+\eps)\log_3 n$ for infinitely many $n$, and thus $\alpha_0\geq 3+\eps$. More generally, Iraids et al.~\cite{iraids2012integer} conjectured that $f(2^i3^j5^k)=2i+3j+5k$ for all $i,j,k$ with $i+j+k>0$ and $k\leq 5$, and Iraids verified this conjecture for all $2^i3^j5^k\leq 10^{12}$ ($k\leq 5$).

We searched for all integers $2^i\leq 2^{126}\approx 10^{38}$, and all $2^i3^j5^k\leq 10^{36}$ ($k\leq 5$). Unfortunately, no counterexample was found.

Iraids et al.~\cite{iraids2012integer} conjectured that $f(2^i+1)=2i+1$ (except 3 and 9). We searched for all integers $2^i+1\leq 2^{121}+1\approx 3\cdot 10^{36}$, and no counterexample was found.

\subsection{Collapse of powers}
J. Arias de Reyna et al.~\cite{de2014question} and Iraids et al.~\cite{iraids2012integer,vcercnenoks2015integer} studied integers whose powers will collapse. We say $n$ \emph{collapses} at $i$, if $i$ is the least integer such that $f(n^i)<i\cdot f(n)$. (If such $i$ does not exist, then we say $n$ is \emph{resistant}.) This concept generalizes the study of $f(2^i)$ in Appendix~\ref{sec:conj_power_of_2}. Finding resistant integers is helpful for proving lower bounds on the constant coefficient $\alpha_0$, since if $n\neq 3$ is resistant, then there exists infinitely many integers $n^i$ such that $\frac{f(n^i)}{\log_3 (n^i)}=\frac{i\cdot f(n)}{i\log_3 n}=
\frac{f(n)}{\log_3 n}>3$. Currently, we only know the trivial bound $\alpha_0\geq 3$.

Iraids et al.~\cite{iraids2012integer} checked the primes within $1000$, and proved that most of them will collapse using the computational results, but there are still a few ones that are not known to collapse at all. Searching within a much larger range $10^{31}$, we found six more primes that will collapse, namely, $733$, $379$, $739$, $541$, $577$, and $811$.

\noindent $f(733^6)=119<120=6\cdot f(733)$, since
\begin{align*}
& 733^6=155104303499468569=(((((((((((((((((((1+1+1)\cdot (1+1+1))\cdot ((1+1+1)\cdot (1+1)))\\
& \cdot (((1+1+1)\cdot (1+1+1))\cdot (1+1+1)))\cdot ((((1+1+1)\cdot (1+1+1))\cdot ((1+1+1)\cdot (1+1))+1)\\
& \cdot (((1+1+1)\cdot (1+1)+1)\cdot (1+1+1)))+1)\cdot (((1+1)\cdot (1+1))\cdot (1+1+1)+1))\cdot ((1+1+1)\\
& \cdot (1+1)+1))\cdot (1+1+1))\cdot (1+1+1))\cdot (1+1))\cdot (1+1))\cdot (1+1))\cdot (1+1))\cdot (1+1)+1)\\
& \cdot (((((1+1+1)\cdot (1+1))\cdot (1+1+1))\cdot ((1+1+1)\cdot (1+1+1))+1)\cdot ((1+1+1)\cdot (1+1+1))\\
& +1))\cdot (((1+1+1)\cdot (1+1))\cdot (1+1+1)+1))\cdot ((1+1+1)\cdot (1+1)+1))\cdot (1+1+1))\\
& \cdot (1+1+1))\cdot (1+1)+1.
\end{align*}
$f(379^6)=107<108=6\cdot f(379)$, since
\begin{align*}
& 379^6=2963706958323721=((((((((((((((((((1+1+1)\cdot (1+1+1))\cdot ((1+1+1)\cdot (1+1+1))+1)\\
& \cdot ((1+1+1)\cdot (1+1))+1)\cdot ((1+1+1)\cdot (1+1+1)))\cdot (((((1+1)\cdot (1+1)+1)\cdot (1+1))\\
& \cdot ((1+1+1)\cdot (1+1))+1)\cdot (((1+1+1)\cdot (1+1+1))\cdot (1+1+1))))\cdot (1+1)+1)\cdot (1+1+1))\\
& \cdot (1+1+1))\cdot (1+1+1))\cdot (1+1)+1)\cdot ((((1+1+1)\cdot (1+1+1))\cdot (1+1+1))\cdot ((1+1+1)\\
& \cdot (1+1+1))+1))\cdot (((1+1+1)\cdot (1+1))\cdot (1+1+1)+1))\cdot ((1+1)\cdot (1+1)+1))\cdot (1+1+1))\\
& \cdot (1+1+1))\cdot (1+1+1))\cdot (1+1+1))\cdot (1+1)+1.
\end{align*}
$f(739^6)=119<120=6\cdot f(739)$, since
\begin{align*}
& 739^6=162879576091729561=((((((((((((((((((1+1+1)\cdot (1+1+1))\cdot (1+1+1))\cdot (((1+1+1)\\
& \cdot (1+1))\cdot (1+1+1)))\cdot (((((1+1)\cdot (1+1))\cdot (1+1+1))\cdot ((1+1+1)\cdot (1+1+1))+1)\cdot ((1+1)\\
& \cdot (1+1)))+1)\cdot ((((1+1+1)\cdot (1+1+1))\cdot (((1+1)\cdot (1+1))\cdot (1+1)))\cdot (((1+1+1)\cdot (1+1+1))\\
& \cdot (((1+1)\cdot (1+1))\cdot (1+1)))+1))\cdot (1+1+1))\cdot (1+1+1))\cdot (1+1)+1)\cdot (1+1+1))\cdot (1+1)\\
& +1)\cdot (((1+1+1)\cdot (1+1+1))\cdot ((1+1+1)\cdot (1+1+1))+1))\cdot (((1+1+1)\cdot (1+1))\cdot ((1+1)\\
& \cdot (1+1)+1)+1))\cdot ((1+1)\cdot (1+1)+1))\cdot (1+1+1))\cdot (1+1+1))\cdot (1+1+1))\cdot (1+1))\cdot (1+1)\\
& +1.
\end{align*}
$f(541^6)=113<114=6\cdot f(541)$, since
\begin{align*}
& 541^6=25071688922457241=((((((((((((((((((((((((((1+1+1)\cdot (1+1+1))\cdot (1+1+1))\\
& \cdot (((1+1+1)\cdot (1+1))\cdot (1+1+1))+1)\cdot (((1+1+1)\cdot (1+1)+1)\cdot ((1+1+1)\cdot (1+1)))+1)\\
& \cdot ((((1+1+1)\cdot (1+1))\cdot ((1+1+1)\cdot (1+1)))\cdot (((1+1+1)\cdot (1+1+1))\cdot (1+1+1))+1))\\
& \cdot (1+1+1))\cdot (1+1+1))\cdot (1+1))\cdot (1+1))\cdot (1+1))\cdot (1+1))\cdot (1+1)+1)\cdot ((1+1)\cdot (1+1)\\
& +1))\cdot ((1+1)\cdot (1+1)+1))\cdot (1+1+1))\cdot (1+1+1))\cdot (1+1+1))\cdot (1+1)+1)\cdot ((1+1)\\
& \cdot (1+1)+1))\cdot (1+1+1))\cdot (1+1+1))\cdot (1+1+1))\cdot (1+1+1))\cdot (1+1))\cdot (1+1))\cdot (1+1)+1.
\end{align*}
$f(577^{12})=227<228=12\cdot f(577)$, since
\begin{align*}
& 577^{12}=1361788799550131972374553991985921=(((((((((((((((((((((((((((((((1+1+1)\cdot (1+1\\
& +1))\cdot ((1+1+1)\cdot (1+1+1))+1)\cdot (((1+1+1)\cdot (1+1+1))\cdot ((1+1+1)\cdot (1+1+1)))+1)\\
& \cdot (((1+1+1)\cdot (1+1))\cdot (1+1+1))+1)\cdot (((((1+1+1)\cdot (1+1))\cdot (1+1+1))\cdot (((1+1)\cdot (1\\
& +1))\cdot ((1+1)\cdot (1+1))))\cdot ((((1+1+1)\cdot (1+1+1))\cdot (1+1+1))\cdot ((1+1+1)\cdot (1+1+1))\\
& +1)+1))\cdot (1+1+1))\cdot (1+1))\cdot (1+1))+1)\cdot (1+1+1))\cdot (1+1+1))\cdot (1+1+1))\cdot (1+1))\\
& +1)\cdot (((((1+1+1)\cdot (1+1))\cdot (1+1+1))\cdot (((1+1+1)\cdot (1+1))\cdot (1+1+1))+1)\cdot (((1+1\\
& +1)\cdot (1+1))\cdot (1+1+1))+1))\cdot (((((((((1+1+1)\cdot (1+1))\cdot (1+1+1))\cdot ((1+1+1)\cdot (1+1\\
& +1)))\cdot ((((1+1+1)\cdot (1+1))\cdot (1+1+1))\cdot ((1+1+1)\cdot (1+1+1)))+1)\cdot (((1+1+1)\cdot (1+1\\
& +1))\cdot ((1+1+1)\cdot (1+1)))+1)\cdot (((((1+1)\cdot (1+1)+1)\cdot (1+1+1))\cdot ((1+1+1)\cdot (1+1+1)))\\
& \cdot (((1+1+1)\cdot (1+1+1))\cdot ((1+1+1)\cdot (1+1+1)))+1))\cdot (1+1+1))+1))\cdot (((1+1+1)\cdot (1\\
& +1))\cdot (1+1+1)+1))\cdot ((1+1+1)\cdot (1+1)+1))\cdot (1+1+1))\cdot (1+1+1))\cdot (1+1+1))\cdot (1\\
& +1))\cdot (1+1))\cdot (1+1))\cdot (1+1))\cdot (1+1))\cdot (1+1))\cdot (1+1))\cdot (1+1))+1).
\end{align*}
$f(811^9)=179<180=9\cdot f(811)$, since
\begin{align*}
& 811^9=151770612880318395249730891=(((((((((((((((((((((((((((((((((((((1+1+1)\cdot(1+1+1))\\
& \cdot(1+1+1))\cdot(((1+1+1)\cdot(1+1))\cdot(1+1+1)))\cdot(((((1+1)\cdot(1+1)+1)\cdot(1+1)+1)\cdot(1+1))\\
& \cdot(((1+1+1)\cdot(1+1))\cdot(1+1+1)))+1)\cdot((((1+1)\cdot(1+1))\cdot(1+1))\cdot((1+1+1)\cdot(1+1)))\\
& +1)\cdot((((1+1+1)\cdot(1+1))\cdot(1+1+1))\cdot((1+1+1)\cdot(1+1+1))+1))\cdot((1+1)\cdot(1+1)+1))\\
& \cdot(1+1+1))\cdot(1+1+1))\cdot(1+1+1))\cdot(1+1+1))\cdot(1+1))\cdot(1+1))\cdot(1+1))\cdot(1+1))\cdot(1+1))\\
& \cdot(1+1))+1)\cdot((1+1)\cdot(1+1)+1))\cdot(1+1+1))\cdot(1+1+1))\cdot(1+1+1))\cdot(1+1+1))\cdot(1+1\\
& +1))\cdot(1+1))+1)\cdot((((((1+1+1)\cdot(1+1))\cdot(1+1+1))\cdot(((1+1)\cdot(1+1)+1)\cdot(1+1+1))\\
& +1)\cdot(1+1))\cdot((((1+1+1)\cdot(1+1+1))\cdot(1+1+1))\cdot(((1+1)\cdot(1+1)+1)\cdot(1+1+1)))+1))\\
& \cdot((1+1)\cdot(1+1)+1))\cdot(1+1+1))\cdot(1+1+1))\cdot(1+1+1))\cdot(1+1+1))\cdot(1+1+1))\cdot(1+1\\
& +1))\cdot(1+1))+1).
\end{align*}

The collapsing status of $109$, $433$, $163$, $487$ and $2$ are still unknown; notice that these are the primes within 1000 that have the smallest logarithmic complexities. We also complete other entries in Table 5 of \cite{iraids2012integer}, by determining that 991 collapses at 6, 257 collapses at 5, and 757 collapses at 6.

\subsection{Estimating the average integer complexity}

One question of interest is to analyze the (asymptotic) average integer complexity. Let $f_{\log}(n)=\frac{f(n)}{\log_3 n}$ denote the \emph{logarithmic integer complexity} of $n$, let $\bar{f}_{\log}(n)=\frac{1}{n-1}\sum_{i=2}^{n} f_{\log}(i)$ denote the average logarithmic integer complexity of the integers less or equal to $n$, and let $\Alphaave=\lim_{n\rightarrow \infty}\bar{f}_{\log}(n)$ denote the asymptotic average logarithmic integer complexity (if it exists). Our goal is to estimate the value of $\Alphaave$.

\begin{table}[!htbp]\centering
\begin{tabular}{|c|c|c||c|c|c|}
\hline
$n$ & average & \# samples & $n$ & average & \# samples\\
\hline
$10^{1}$ & 3.238373 & $10^{6}$ & $10^{14}$ & 3.342084 & $5\cdot 10^{5}$ \\
\hline
$10^{2}$ & 3.349894 & $10^{6}$ & $10^{15}$ & 3.338514 & $5\cdot 10^{5}$ \\
\hline
$10^{3}$ & 3.393001 & $10^{6}$ & $10^{16}$ & 3.335244 & $3\cdot 10^{5}$ \\
\hline
$10^{4}$ & 3.400376 & $10^{6}$ & $10^{17}$ & 3.332222 & $2\cdot 10^{5}$ \\
\hline
$10^{5}$ & 3.395626 & $10^{6}$ & $10^{18}$ & 3.329551 & $10^{5}$ \\
\hline
$10^{6}$ & 3.388161 & $10^{6}$ & $10^{19}$ & 3.327147 & $10^{5}$ \\
\hline
$10^{7}$ & 3.380561 & $10^{6}$ & $10^{20}$ & 3.324807 & $5\cdot 10^{4}$ \\
\hline
$10^{8}$ & 3.373263 & $10^{6}$ & $10^{21}$ & 3.322698 & $5\cdot 10^{4}$ \\
\hline
$10^{9}$ & 3.366459 & $10^{6}$ & $10^{22}$ & 3.320632 & $2\cdot 10^{4}$ \\
\hline
$10^{10}$ & 3.360493 & $10^{6}$ & $10^{23}$ & 3.318700 & $2\cdot 10^{4}$ \\
\hline
$10^{11}$ & 3.355246 & $10^{6}$ & $10^{24}$ & 3.316791 & $2\cdot 10^{4}$ \\
\hline
$10^{12}$ & 3.350321 & $5\cdot 10^{5}$ & $10^{25}$ & 3.315385 & 2000 \\
\hline
$10^{13}$ & 3.346080 & $5\cdot 10^{5}$ & & & \\
\hline
\end{tabular}
\caption{Estimating the average logarithmic integer complexity $\bar{f}_{\log}(n)$ for $n\leq 10^{25}$, by random sampling.}
\label{table:sampling}
\end{table}

\paragraph{Computing $\bar{f}_{\log}(n)$.} J. Arias de Reyna et al.~\cite{de2014question} have computed $\bar{f}_{\log}(n)$ for $n\leq 905000000$, and get $\bar{f}_{\log}(905000000)\approx 3.366$. They directly used this value as an estimation of $\Alphaave$, however, we point out that $\bar{f}_{\log}(n)$ seem to decrease when $n$ increases. A more precise estimation of $\Alphaave$ requires computing $\bar{f}_{\log}(n)$ for larger $n$, which is a challenging task for the previous superlinear-time algorithms. With our new single-target algorithm, we are able to approximately compute $\bar{f}_{\log}(n)$ for a wider range of $n$, by \emph{random sampling}.

Our computational results are shown in Table~\ref{table:sampling}. the $99.999\%$ confidence intervals have length $\leq 0.0013$ for $n\leq 10^{24}$, and length $\leq 0.0034$ for $n\leq 10^{25}$.

\paragraph{Predicting $\Alphaave$.}




We try to fit the data in Table~\ref{table:sampling} by curves, and the result is shown in Fig.~\ref{fig:fit}. We empirically use the model $\bar{f}_{\log}(n)=\frac{a}{x^2+bx+c}+d$ with parameters $a,b,c,d$ and $x=\log_{10} n$, and drop the first few data points (since the asymptotic behaviour may not be reflected when $n$ is too small). It seems that the limit of $\bar{f}_{\log}(n)$ exists and $\Alphaave\approx 3.29$, which is close to the previous ones reported in \cite{cordwell2017algorithms,iraids2012integer} (however, our new method seem to have higher precision).


\begin{figure}[!htbp]\centering
\includegraphics[width=.6\textwidth]{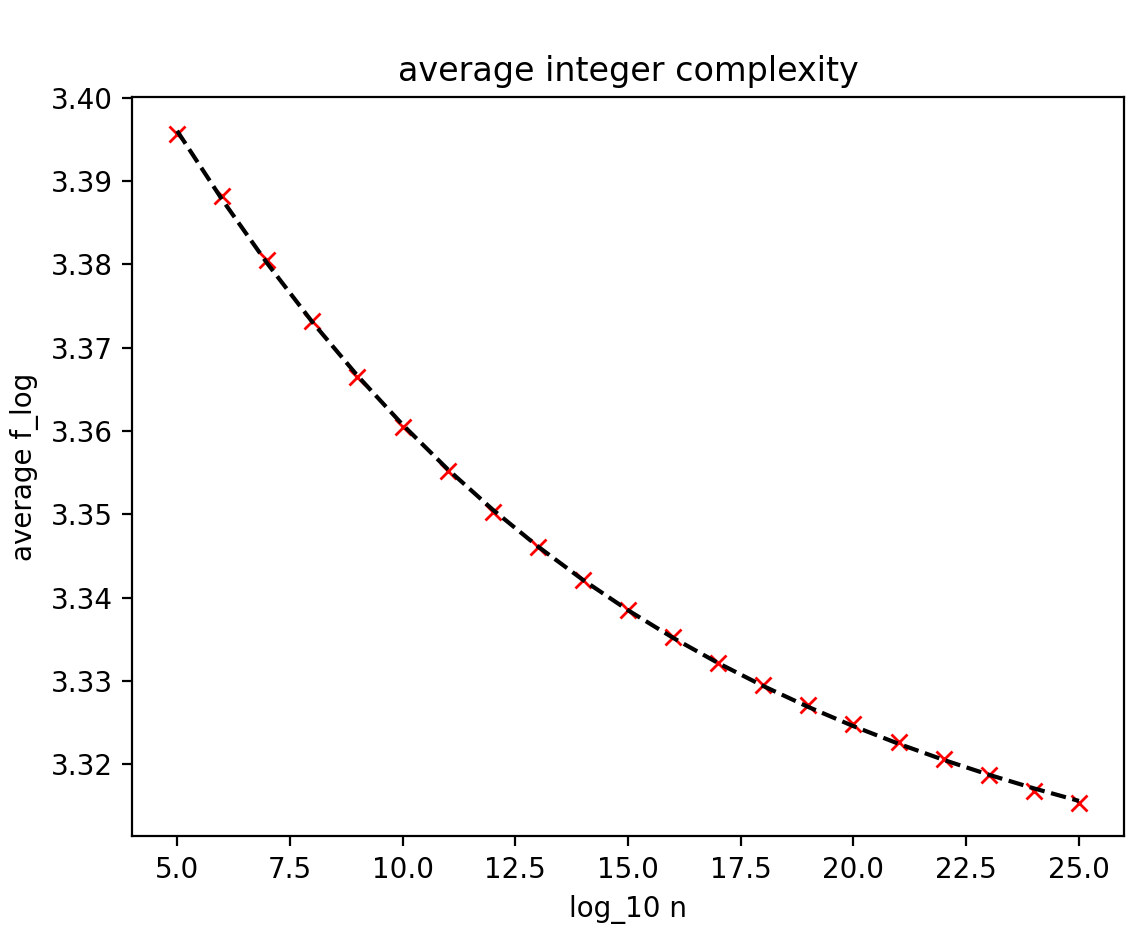}\\
\caption{Curve fitting for $\bar{f}_{\log}(n)$. Here the $x$-coordinate is $\log_{10} n$, and the $y$-coordinate is the computed approximation of $\bar{f}_{\log}(n)$. Empirically using the model $y=\frac{a}{x^2+bx+c}+d$, the best parameter is $a=28.3772$, $b=11.5122$, $c=184.398$, and $d=3.28974$.}\label{fig:fit}
\end{figure}

\end{document}